\documentclass[letterpaper, 10 pt, journal]{ieeeconf}  

\IEEEoverridecommandlockouts                              
\usepackage{amsmath} 
\usepackage{amssymb}  

\usepackage{graphicx}

\newcommand{\T}{T}

\newcommand{\eps}{\epsilon}

\newcommand{\E}{\mathbb{E}}

\newcommand{\R}{\mathbb{R}}

\newcommand{\Bc}{\mathcal{B}}

\newcommand{\Dc}{\mathcal{D}}

\newcommand{\Mc}{\mathcal{M}}

\newcommand{\Sc}{\mathcal{S}}

\newcommand{\diag}{\mathrm{diag}}

\usepackage{amsthm}

\newtheorem{assumption}{Assumption}
\newtheorem{definition}{Definition}
\newtheorem{proposition}{Proposition}
\newtheorem{lemma}{Lemma}
\newtheorem{theorem}{Theorem}

\title{\LARGE \bf
Statistical Estimation with Strategic Data Sources in Competitive Settings
}

\author{Tyler Westenbroek, Roy Dong, Lillian J. Ratliff, and S. Shankar Sastry
\thanks{T. Westenbroek, R. Dong, and S. S. Sastry are with the Department of Electrical Engineering and Computer Sciences, University of California, Berkeley, Berkeley, CA, 94707, USA, {\tt\footnotesize $\{$westenbroekt,roydong,sastry$\}$@eecs.berkeley.edu}. L. J. Ratliff is with the Department of Electrical Engineering, University of Washington, Seattle, WA, 98195, USA, {\tt\footnotesize ratliffl@uw.edu}. }%
}

\begin{document}

\maketitle
\thispagestyle{empty}
\pagestyle{empty}


\begin{abstract}

In this paper, we introduce a preliminary model for interactions in the data market. Recent research has shown ways in which a data aggregator can design mechanisms for users to ensure the quality of data, even in situations where the users are effort-averse (i.e. prefer to submit lower-quality estimates) and the data aggregator cannot observe the effort exerted by the users (i.e. the contract suffers from the principal-agent problem). However, we have shown that these mechanisms often break down in more realistic models, where multiple data aggregators are in competition. Under minor assumptions on the properties of the statistical estimators in use by data aggregators, we show that there is either no Nash equilibrium, or there is an infinite number of Nash equilibrium. In the latter case, there is a fundamental ambiguity in who bears the burden of incentivizing different data sources. We are also able to calculate the price of anarchy, which measures how much social welfare is lost between the Nash equilibrium and the social optimum, i.e. between non-cooperative strategic play and cooperation.

\end{abstract}

\section{Introduction}
\label{sec:intro}

The proliferation of smart sensors in recent years has introduced the possibility of accurately detecting and estimating a large new class of phenomena that affect society. These sensors, ranging from smart personal devices to more traditional purpose-built sensors, may be owned by a multitude of sources, and can produce qualitatively different readings which can be combined to make inferences about an event of interest.

In turn, this has led to the advent of crowd sensing, wherein a central data collector accrues the measurements made by a multitude sources, using these data points to generate a single cohesive estimate for some phenomena of interest to the data collector. 
However, the quality of this central estimate, and thus its value to the data collector, depends fundamentally on the ability, and moreover the willingness, of the data sources to produce accurate readings which are relevant to the phenomena the data collector wishes to study. 

Unfortunately, there may be instances where data sources have some aversion to
providing the data collector with the quality of estimates she desires. Take as
an example, the case where the sensor must exert significant resources to
produce an accurate reading (e.g. time or network bandwidth), or a situation
where the source views the information she is sharing as private, and has
incentive to obfuscate the data she shares~\cite{Bakken2004,Dwork2014}. Consequently,  in order to ensure she consistently receives high quality measurements from the data sources, the central data collector must design an incentive mechanism which:

\begin{enumerate}
\item
allows her to metricize the quality of the reading each data source provides, and
\item
provides incentive for the data sources to produce readings which are considered "high quality" under this metric.
\end{enumerate}

Given the wide range of applications and industries this problem affects, many different compensation mechanisms have been proposed to promote the production of high quality readings from a collection of data sources. An overview of such mechanisms is given in~\cite{gao2015survey}. 

The contribution of this section can be seen as an extension of \cite{cai:2015aa},
in which the authors design a general payment mechanism, by which a central data
collector may induce each data source in the marketplace to exert precisely the
level of effort in collecting data that the central data buyer desires. The goal
of the data buyer in this case is to obtain a high quality estimator for some
phenomena using the readings from the data sources, while reducing the payments
needed to incentivize the necessary exertion of effort from the sensors. Several
other papers~\cite{dobakhshari2016incentive,farokhi2015budget} further investigate mechanisms of this sort, proposing several extensions. 

However, it has yet to be studied how such mechanisms perform in situations where more than one central data buyer wishes to purchase readings from data sources in the marketplace. A number of important questions arise when such \emph{data markets} are considered. If the central data buyers are competing companies, will they permit data sources to also sell information to their competitors? If the data buyers do purchase readings from the same set of data sources, who will foot the bill to incentivize the effort the data sources exert? Will the data buyers who provide larger payments to the data sources be compensated with higher quality readings than their competitors?

Most significantly, this section demonstrates that if all the data buyers design compensation schemes as proposed in \cite{cai:2015aa}, each of the data buyers will receive the same quality of reading from a particular data source, regardless of how much each data buyer personally compensates the data source for her effort. This leads to conflicting objectives for each of the data buyers on several fronts. If a data buyer wishes to induce a data source to exert a high level of effort, she must reconcile the fact that her competitors will also receive a high quality reading from this data source. Even in the case where the data buyers care little about the success of the other buyers in the marketplace, each data buyer still wants to incentivize the data sources to produce high quality readings, but wants to force the other data buyers to offer the lion's share of the necessary compensation. 

In this section, we analyze the competitive outcomes that arise in such a marketplace by formulating a game between the buyers wherein they

\begin{enumerate}
\item compete by designing pricing mechanisms to affect the behavior of the data
    sources, and
\item design these mechanisms so as meet the personal objectives enumerated
    above.
\end{enumerate}

We derive conditions for the existence of Nash Equilibria in this game when a particular form is assumed, and analyze the efficiency and equity of these outcomes.  We demonstrate through both analytical and numerical exercises that the outcomes of these games are often highly inefficient from a social standpoint, which motivates future work  to design  incentive mechanisms which more effectively handle competition between data buyers.

The rest of this paper proceeds as follows. In Section~\ref{sec:datam_formulation}, we lay out explicit mathematical structures for the data markets, strategic data sources, strategic data buyers, and the class of contracts we will consider between the sources and buyers. In Section~\ref{sec:datam_results}, we analyze the game that forms between the buyers in the data market, and demonstrate that the outcome of this game is in many cases socially inefficient, and often times. Section~\ref{sec:datam_examples} provides a numerical example which highlights the issues  presented in Section~\ref{sec:datam_results}. And finally, Section~\ref{sec:datam_discussion} prescribes  an agenda for future work, with the aim of developing more refined incentive mechanism which do  not suffer from the same shortcomings in the competitive setting.

\section{Mathematical formulation}
\label{sec:datam_formulation}

In this section we formulate our model for data markets. We first present our model for strategic data sources, and then strategic buyers who issue incentives to strategic data sources. Based on recent research~\cite{cai:2015aa}, we use incentives with a particular payment structure. Then, we define our overall game, as well as a generalized Nash equilibrium for this game.

\subsection{Data market}

At a high level, a \emph{data market} consists of a set $\mathcal{S} = \{1, ... , N\}$ of strategic data sources, and a set $\mathcal{B} = \{1, ... , M\}$ of strategic data buyers. Each data source $i$ is equipped to generate an estimate of the function $f : \mathcal{D} \to \R$ at some data point $x_i \in \mathcal{D}$, and each data buyer $j \in \mathcal{B}$ wishes to use these readings to generate a personal estimator of $f$, which we will denote $\hat{f}^j$. Each buyer $b_j$ is willing to form a contract with each data source $i \in \mathcal{S}$, which monetarily compensates $i$ for the readings she produces, and we assume it is under the purview of $j$ to define the structure of this contract.

One may think of $\mathcal{D}$ as a set of features or events the data buyers are capable of observing, in order to make a prediction about some phenomena. The value returned by the mapping $f$ encapsulates the relationship between the observable features and the outcome of interest. We further assume that each of the data sources and buyers acts \emph{strategically}; that is, each of these agents acts to maximize some expected personal return from her transactions in this marketplace. The following two subsections of the document provide an explicit mathematical formulation describing the behavior of the data sources and data buyers. The basis for these definitions comes directly from~\cite{cai:2015aa}.

\subsection{Strategic data sources}

In this subsection, we define our model for strategic data sources. Intuitively, data sources provide data samples $(x,y)$ whose variance depends on their effort. Thus, the more effort exerted, the better the statistical estimation for any data buyer who receives the data. Additionally, we assume the data sources are effort-averse, i.e. all else equal, they prefer to exert minimal effort. Furthermore, the buyer has no direct way to verify the amount of effort exerted by the data source. Thus, we have an issue commonly referred to as \emph{moral hazard}.

More formally, all data sources share some function $f : \Dc \rightarrow \R$, where $f$ is the function which data buyers wish to estimate. One may think of $\Dc$ as a set of features or events the data buyers are capable of observing, in order to make a prediction about some phenomena. The value returned by the mapping $f$ encapsulates the relationship between the observable features and the outcome of interest. 

Each data source $i$ has their own feature $x_i \in \Dc$ and their own cost-of-effort function $\sigma_i^2 : \R \rightarrow \R_+$. When data source $i$ exerts effort $e_i \in \R$, they produce an estimate of the form:
\[
y_i(e_i) = f(x_i) + \eps_i(e_i) \quad \eps_i(e_i) \sim N(0,\sigma_i^2(e_i))
\]
Both $x_i$ and $\sigma_i^2$ are common knowledge, but the effort $e_i$ is private, as well as the the value $y_i(e_i)$ produced. We shall design contracts such that the data source $i$ is incentivized to exert the `correct' amount of effort (to be defined), and report $y_i$ truthfully.

Data source $i$ will receive a payment from each buyer for their data. For buyer $j$, let this payment, potentially random, be denoted $p_i^j$. We assume that the data source has a utility function of the following form, should they opt-in:
\begin{equation}
\label{eq:effort_selection}
\E  \left( \sum_{j \in \Bc} p_i^j \right) - e_i 
\end{equation}
If they opt-out, they will receive utility 0. 

Note that this assumes that the data sources are risk-neutral, effort-averse, and must opt-in ex-ante. Additionally, we assume the effort $e_i$ can be normalized to be comparable to the payments.

Throughout the rest of this paper, we shall often omit the argument $e_i$ when context makes it evident. 

\subsection{Strategic data buyers}

A strategic data buyer $j \in \Bc$ is an agent who wishes to construct the best estimator $\hat{f}^j$ for a function $f$. She optimizes a loss function across a class of estimators, which the data buyer is free to select. In general, different buyers need not fit models of the same type; for example, one data buyer may choose to generate her estimator via linear regression, while another data buyer constructs his estimator by fitting the data to a polynomial model of higher degree. Differences in the type of estimator data buyers use may be used to encapsulate competitive advantages one data buyer has over another. For a more thorough review of the technical requirements of these estimators, see~\cite{cai:2015aa}.

Additionally, each data buyer $j$ has a distribution $F_j$ across $\Dc$, which denotes how much they value an accurate estimate at various points in $\Dc$.

In particular, let $\hat{f}^j_{(\vec{x}, \ \vec{y^j})}$ denote the estimator that buyer $j$ constructs, based on the location of the data sources, $\vec{x}$,  and the reports she receives from the data sources, $\vec{y^j}$. (Here, $\vec{x} = (x_1,\dots,x_N)$ and similarly $\vec{y^j}$ is the vector of $y$ values reported to buyer $j$.)

Beyond any intrinsic utility buyer $j$ experiences from increasing the quality of her estimator, $j$ also wishes to construct an estimator that is better than the estimator constructed by her competitors, the other members of $\mathcal{B}$.

Each data buyer $j$ commits to a payment function $p_i^j$ to each data source $i \in \mathcal{S}$, where $p_i^j: \mathcal{D}^N \times \R^N \to \R$ may depend not only on the reading reported by data source $i$, but also the readings reported by the other members of $\mathcal{S}$, with consideration given to the location of the data sources. In particular, buyer $j$ constructs her various contracts with the data sources so as to minimize:
\begin{eqnarray}
\label{eq:buyer_obj}
\begin{split}
J^j&(\vec{p^j}, \vec{p^{-j}})  =
\E \Bigg[ 
\left( \hat{f}^j_{\vec{x},\vec{y^j}}(x^*) - f(x^*) \right)^2 - \\
&
\sum_{k \in -j} \delta_k^j \left( \hat{f}^k_{\vec{x},\vec{y^k}}(x^*) - f(x^*) \right)^2 +
\eta^j \sum_{i \in \Sc} p_i^j(\vec{x},\vec{y^j})
\Bigg]
\end{split}
\end{eqnarray}
The expectation in \eqref{eq:buyer_obj} is taken across $x^* \sim F_j$ as well as the randomness in the reported data $\vec{y^k}$ for $k \in \Bc$. (Recall that $F_j$ weighs the importance data buyer $j$ places on an an accurate estimator about different points $x^* \in \mathcal{D}$.) 
Here, as per typical game theory notation, we will let $-i$ denote $\Sc \setminus \{i\}$ and $-j$ denote $\Bc \setminus \{j\}$, and when $-i$ or $-j$ is used as a subscript, this denotes everyone else's variables, e.g. $\vec{p^{-j}}$ denotes the vector of payment plans of all the data buyers that are not $j$.

Here, $\delta_k^j \in [0 ,1]$ parameterizes the level of competition between buyers $j$ and $k$, and we assume this competition is symmetric so $\delta_k^j = \delta_j^k$. When $\delta_k^j =0$, $j$ is indifferent to the success of $k$, and competes with $k$ only insofar as trying to determine who will pay to incentive the data sources. Meanwhile, $\delta^j_k  = 1$ denotes a situation akin to a zero-sum game between data buyers $j$ and $k$. 

The parameter $\eta^j > 0$ denotes a conversion between dollar amounts allocated by the payment functions and the utility generated by the quality of the various estimators that are constructed.

In order for the objective expressed in \eqref{eq:buyer_obj} to be well defined, we assume that buyer $j$ chooses to construct an estimator for which there exists a function $g_j$ such that, for all distributions $F^j$ over $\mathcal{D}$, $\vec{x}$, and $\vec{\sigma^2} \in \R^N$:
\begin{equation}
g_j(\vec{x}, F_j, \vec{\sigma^2}) = \E \left[ \left( \hat{f}^j_{\vec{x},\vec{y^j}}(x^*) - f(x^*) \right)^2 \right]
\end{equation} 
Here the $\vec{y^j}$ have variance $\vec{\sigma^2}$.

Finally, we assume that buyer $j$ has knowledge of what class of estimator each of
the other data buyers plans to use.\footnote{This is a heavy-handed assumption,
given that competing data buyers are unlikely to inform their competitors how
they intend to process the data supplied by the sources. However, this is
keeping with the goal of the paper, as we shall demonstrate that even when there
is complete information between the buyers, inefficiencies still arise in the
data market.} 

The data buyers are interested in offering payment contracts to data sources. These contracts must be designed such that, for each data source $i$, when $i$ selects her effort $e_i$ to maximize to \eqref{eq:effort_selection}, given the payment contracts from all of the other
data buyers:
\begin{equation}
\label{eq:total_IR}
\E \left[ \sum_{j \in \Bc} p_i^j(\vec{y^j}(e_i))\right] - e_i \geq 0
\end{equation}
\begin{equation}
\label{eq:single_IR}
\E \left[ p_i^j(\vec{y^j}(e_i)) \right] \geq 0
\end{equation}
Note that \eqref{eq:total_IR} is an ex-ante constraint for data source $i$ that $i$ receives non-negative payoff in expectation. This depends on the payments of the other data buyers. The second is an ex-ante constraint that data source $i$ never opts into any contract with negative payments.

We model the resulting competition between the data buyers, subject to these coupled constraints, as a generalized Nash equilibrium problem (GNEP)~\cite{dorsch:2013aa}.
\begin{definition}
\label{def:gnep}
Each player $j$ from a finite set of players $\Bc$ aims to solve an optimization problem given by:
\begin{equation}
\label{eq:gnep1}
BR(p^{-j}) = \arg\min_{p^j} \{ J^j ( p^j, p^{-j} ) \vert p^j \in \Mc^j(p^{-j}) \}
\end{equation}
$\Mc^j(p^{-j})$ is called the \emph{feasible set} for player $j$, which depends on the actions taken by the other players $-j$. A vector $p = (p^1, p^2, \dots p^M)$ is called a \emph{(generalized) Nash equilibrium (GNE)} if $p^j = BR(p^{-j})$ for all $j \in \Bc$, i.e. the $p^j$ are simultaneously solutions to each players optimization \eqref{eq:gnep1}.
\end{definition}

Having laid out the general formulation for this problem, in the final portion of this paper we lay out the form of the payment contracts that we consider between the buyers and sellers.

\subsection{Structure of payment contracts}

In~\cite{cai:2015aa} the particular case where $|\mathcal{B}| = 1$ is analyzed, and no competition between buyers of data must be considered. Their work considers payment plans from the single buyer to each data source $i$ of the form:
\begin{equation}
p_i(\vec{x} , \vec{y}) = c_i - d_i \left( y_i - \hat{f}_{(\vec{x}, \vec{y^j})_{-i}}(x_i) \right)^2,
\end{equation}
where $\hat{f}_{(\vec{x}, \vec{y})_{-i}}(x_i)$, is the optimal estimate for $f(x_i)$ that the data buyer can construct from the readings reported by the data sources other than source $i$, and $c_i \geq 0, d_i \geq 0$ are scalars to be chosen strategically by the buyer. The authors of~\cite{cai:2015aa} demonstrate an algorithm for selecting $c_i$ and $d_i$ which allows the buyer to:
\begin{enumerate}
\item precisely incentive data source $i$ to exert any level of effort $\bar e_i$ that the buyer desires (the authors can make $\bar e_i$ a dominant strategy for data source $i$), and 
\item precisely compensate data source $i$ for her effort ($\E p_i(y_i(\bar e_i), \vec{y}_{-i}(\vec{\bar{e}}_{-i})) = e_i$, making the contract tightly satisfy individual rationality constraints).
\end{enumerate}

Our goal is to study how pricing schemes of this form perform in the more general case where $|\mathcal{B}| > 1$, and competition between multiple data buyers becomes a critical consideration. In particular,  we assume the following form for each of the incentive mechanisms offered in the data market.
\begin{assumption}
\label{ass:price_structure}
Consider a data buyer $j$ and data source $i$.  It is assumed that $j$ offers $i$ a payment function of the form 
\begin{equation}
p_i^j(\vec{x}, \vec{y}) = c_i^j - d_i^j \left( y_i - \hat{f}^j_{(\vec{x}, \vec{y})_{-i}}(x_i) \right)^2,
\end{equation}
in exchange for knowledge of  $y_i$,where $c_i^j ,d_i^j\geq 0$ are parameters that the buyer $j$ is free to choose.
\end{assumption}
Note that these payments do not directly depend on the level of effort that any of the data sources exert, since the data buyers do not have a means to directly observe these values. The payments only depend on the data reported to them, and can be calculated by data buyers.
Having defined the necessary structures for the data markets we wish to study, we are now ready to study the competitive equilibria that arise in these marketplaces.

First, we note that for any data source $i$, due to the form of the payment contract, they will report the same value to all data buyers. 

\begin{proposition}
Fix any data source $i$. Pick any vector of variances $\vec{\sigma^2}$ (one variance for each data buyer), and let $e = \max~\{ \tilde e : \sigma_i^2(\tilde e) = (\vec{\sigma^2})_j \}$, i.e. $e$ is the minimum amount of effort for data source $i$ to generate measurements of variance $\vec{\sigma^2}$. Then, data source $i$ has higher payoff, defined by \eqref{eq:effort_selection}, by choosing variances $\sigma_i^2(e)$ for all $j$, than the payoff earned from providing each buyer $j$ with data of variance $(\vec{\sigma^2})_j$.
\end{proposition}
In other words, since the payment contract from each data buyer $j$ is increasing (in expectation) with respect to effort, data source $i$ will never have incentive to `add noise' to a measurement once the effort has been exerted. Thus, for the rest of this paper, we shall write $\vec{y}$ to denote the measurement reported to all data sources $j$.

\section{Results}
\label{sec:datam_results}

In this section, we analyze the behavior we can expect from each of the agents in the market place, by considering the game that forms between the members of $\mathcal{B}$ as they select the parameters in the contracts they offer to the data sources. 

Adopting standard game-theoretic short-hand notion, we denote the set of pricing
parameters buyer $k$ selects by $(c^k, d^k)$ , and we denote the choice of the
pricing parameters of the other members of $\mathcal{B}$ by $(c^{-k}, d^{-k})$. 
From now on, we use the index $k$ to single out a specific buyer, the index $q$ to single out a data source, the index $j$ to sum over a collection of buyers, and the index $i$ (and sometimes $l$) to sum over a collection of sources. 

We begin our analysis by determining under what conditions the data sources will accept the collection of contracts offered to them by the data buyers. 
Recall that data source $q$ will accept all of the contracts offered by the data sources if and only if the ex-ante total payments are non-negative \eqref{eq:total_IR} and each data buyer's payment is non-negative ex-ante \eqref{eq:single_IR}.

Let $\delta_x$ denote the probability measure that puts mass 1 at point $x$. Then, we may simplify \eqref{eq:total_IR} for a fixed $q$ by noting that:
\[
\E \left[ \sum_{j \in \Bc} p_q^j (\vec{x}, \vec{y} ) \right] =
\sum_{j \in \Bc} c_q^j - \E \sum_{j \in \Bc} d_q^j\left( y_q - \hat{f}^j_{\vec{x}_{-q},\vec{y}_{-q}}(x_q) \right)^2 =
\]
\[
\sum_{j \in \Bc} c_q^j - \sum_{j \in \Bc} d_q^j\left( \sigma_q^2(e_q) + g_j(\bar{x}_{-q}, \delta_{x_q}, \vec{\sigma^2}_{-q}) \right)
\]
Then, \eqref{eq:total_IR} holds if and only if:
\begin{equation}
\label{eq:IR_conditions}
\sum_{j \in \Bc} c_q^j - \sum_{j \in \Bc} d_q^j\left( \sigma_q^2(e_q) + g_j(\bar{x}_{-q}, \delta_{x_q}, \vec{\sigma^2}_{-q}) \right) \geq e_q
\end{equation}
Similarly, \eqref{eq:single_IR} holds if and only if:
\begin{equation}
\label{eq:IR_conditions2}
c_q^j \geq d_q^j \left( \sigma_q^2(e_q) + g_j(\bar{x}_{-q}, \delta_{x_q}, \vec{\sigma^2}_{-q}) \right)
\end{equation}

As our goal is to find situations where the buyers receive data from each of the data sources, we shall include equations \eqref{eq:IR_conditions} and \eqref{eq:IR_conditions2} as constraints in the game between data buyers. Indeed, given a choice of $(c^{-k}, d^{-k})$, the objective of buyer $k$ is to optimize the following problem:
\begin{eqnarray}
\min_{c^k,d^k} & J^k( (c^k, d^k), (c^{-k}, d^{-k}) ) \\
\text{s.t.} & \E\left[ \sum_{j \in \Bc} p_i^j( \vec{x},\vec{y}(\vec{e^*}) ) \right] - e_i^* \geq 0 \\
\label{eq:source_opt_constraint}
& e_i^* = \arg\max_{e_i} \E \left[ \sum_{j \in \Bc} p_i^j (\vec{x}, \vec{y}(\vec{e^*})) \right] - e_i  \\
& \E \left[ p_i^k(\vec{x},\vec{y}(\vec{e} ) ) \right] \geq 0 \\
& c_i^k \geq 0, d_i^k \geq 0
\end{eqnarray}
Each constraint holds for all $i \in \Sc$. 
Recall that $J^k$ was defined in \eqref{eq:buyer_obj}. Note that~\cite{cai:2015aa} showed that the payments induce dominant strategies, so \eqref{eq:source_opt_constraint} is an optimization that does not depend on $e_{-i}$.

In general, this may be a computationally difficult problem for $b_k$ to solve. For illustrative purposes, for the rest of this paper, we will assume specific forms for the estimators the buyers employ and the $\sigma$ functions which define the data sources. We first assume:

\begin{assumption}
\label{ass:sigma_form}
For each data source $i$,  $\sigma_i(e_i)$ is characterized by the the constant $\alpha_i > 0$ and of the form:
\begin{equation}
\sigma_i(e_i) = \exp(-\alpha_i e_i)
\end{equation}
\end{assumption}
Note that this implies that $\sigma$ is convex, strictly decreasing and always positive, which are all desirable properties in our context. Furthermore, note that this is the form of the standard deviation, not the variance.

We next determine the level of effort data sources will exert given the pricing parameters set by the data buyers.
Fix a data source $q$ and taking the derivative of \eqref{eq:effort_selection} with respect to $e_q$, we obtain:
\[
-2 \left( \sum_{j \in \Bc} d_q^j \right) \sigma_i(e_q) \frac{d}{de_q} \sigma_q(e_q) -1 =  
\]
\[
2  \left( \sum_{j \in \Bc}d_q^j\right) \alpha_q \exp(-2 \alpha_q e_q) - 1
\]
Setting this derivative equal to 0 yields:
\begin{equation}
\label{eq:specific_effort}
e_q^* = \frac{\ln\left(2 \left( \sum_{j \in \Bc}d_q^j\right) \alpha_q\right)}{2 \alpha_q}
\end{equation}
This is the optimum effort selection for data source $q$. We can also compute how this optimal point varies with $d_i^j$:
\[
\frac{\partial}{\partial d_q^j} e_q^* = \frac{1}{2 \left( \sum_{j \in \Bc}d_q^j\right)\alpha_q }
\]
Also we can easily calculate the optimum variance:
\begin{equation}
\label{eq:opt_variance_q}
\sigma_q^2(e_q^*) = \frac{1}{2 \left( \sum_{j \in \Bc}d_q^j\right) \alpha_q}
\end{equation}
\begin{assumption}\emph{(Separable estimators)}
\label{ass:separable}
For each buyer $k \in \Bc$, the estimator for $f$ that buyer $k$ employs, $\hat f^k$, is separable. In other words, there exists a function $h_k$ such that:
\[
g_k(\vec{x}, F, \vec{\sigma^2}) = \sum_{i \in \Sc} h_k(x_i, \vec{x}, F) \sigma_i^2
\]
Furthermore, we assume that $h \geq 0$.
\end{assumption}
Note that  linear regression, polynomial regression and finite-kernel regression all produce separable estimators. Applying Assumption~\ref{ass:separable} for the estimators, we may rewrite the loss function for buyer $k$ as:
\begin{eqnarray*}
J^k  ( (c^k,d^k), (c^{-k},d^{-k}) ) = 
 \sum_{i \in \Sc} h_k(x_i, \vec{x}, F_k) \sigma_i^2(e_i^*) - \\
  \sum_{j \in -k} \delta_j^k \sum_{i \in \Sc} h_j (x_i, \vec{x}, F_j) \sigma_i^2(e_i^*) + \\
  \eta^k \sum_{i \in \Sc} \left( c_i^k - d_i^k \left[ \sigma_i^2(e_i^*) + \sum_{l \in -i} h_k(x_l, \vec{x}_{-i}, \delta_{x_i}) \sigma_l^2(e_l^*) \right] \right) 
\end{eqnarray*}
Recall that each $x_i$ is fixed and common knowledge; thus, we can replace each of the above evaluations of the $h$ functions with constants. Define $\beta_i^j = h_j(x_i, \vec{x}, F_j)$, $\xi_{i,l}^j = h_j(x_l, \vec{x}_{-i}, \delta_{x_i})$ for $i \neq l$ and $\xi_{i,i}^j = 1$. Note that $\xi \geq 0$. Then, this becomes:
\begin{eqnarray*}
J^k( (c^k,d^k), (c^{-k},d^{-k}) ) = 
\sum_{i \in \Sc} \beta_i^k \sigma_i^2(e_i^*) - \\
 \sum_{j \in -k} \delta_j^k \sum_{i \in \Sc} \beta_i^j \sigma_i^2(e_i^*) + \\
\eta^k \sum_{i \in \Sc} \left( c_i^k - d_i^k \left[ \sigma_i^2(e_i^*) + \sum_{l \in -i} \xi_{i,l}^k \sigma_l^2(e_l^*) \right] \right) 
=
\\
\sum_{i \in \Sc} \left( \beta_i^k - \sum_{j \in -k} \delta_j^k \beta_i^j \right) \sigma_i^2(e_i^*) + \\
\eta^k \sum_{i \in \Sc} \left( c_i^k - d_i^k \left[ \sigma_i^2(e_i^*) + \sum_{l \in -i} \xi_{i,l}^k \sigma_l^2(e_l^*) \right] \right) 
\end{eqnarray*}
In efforts towards succinctness, let $\gamma_i^k = \beta_i^k - \sum_{j \in -k} \delta_j^k \beta_i^j$. We will now plug in the expression for $\sigma_i^2(e_i^*)$ in \eqref{eq:opt_variance_q}, yielding:
\begin{eqnarray*}
J^k( (c^k,d^k), (c^{-k},d^{-k}) ) = 
\sum_{i \in \Sc} \gamma_i^k \sigma_i^2(e_i^*) + \\ 
\eta^k \sum_{i \in \Sc} \left( c_i^k - d_i^k \left[ \sigma_i^2(e_i^*) + \sum_{l \in -i} \xi_{i,l}^k \sigma_l^2(e_l^*) \right] \right) 
 =
 \\
\sum_{i \in \Sc} \frac{ \gamma_i^k }{2 \left( \sum_{j \in \Bc}d_i^j\right) \alpha_i} + 
\\
\eta^k \sum_{i \in \Sc} \Bigg( c_i^k - d_i^k \Bigg[ \frac{ 1 }{2 \left( \sum_{j \in \Bc}d_i^j\right) \alpha_i} + \\
 \sum_{l \in -i}  \frac{ \xi_{i,l}^k }{2 \left( \sum_{j \in \Bc}d_l^j\right) \alpha_l} \Bigg] \Bigg) 
=
\\
\sum_{i \in \Sc} \frac{ \gamma_i^k }{2 \left( \sum_{j \in \Bc}d_i^j\right) \alpha_i} + \\
\eta^k \sum_{i \in \Sc} \left( c_i^k - d_i^k \left[ \sum_{l \in \Sc} \frac{ \xi_{i,l}^k }{2 \left( \sum_{j \in \Bc}d_l^j\right) \alpha_l} \right] \right) 
\end{eqnarray*}
(Note here we joyfully take advantage of our convention that $\xi_{i,i}^k = 1$.)

Finally, similar reasoning lets us write for any data source $q$ and data buyer $k$:
\[
\E \left[ p_q^k (\vec{x}, \vec{y} ) \right] =
c_q^k - d_q^k\left( \sigma_q^2(e_q) + g_k(\bar{x}_{-q}, \delta_{x_q}, \vec{\sigma^2}_{-q}) \right) =
\]
\[
c_q^k - d_q^k\left( \sigma_q^2(e_q) + \sum_{i \in -q} h_k(x_i, \vec{x}_{-q}, \delta_{x_i}) \sigma_i^2(e_i) \right) =
\]
\[
c_q^k - d_q^k \left( \sum_{i \in \Sc} \xi_{q,i}^k \sigma_i^2(e_i) \right) 
\]
At optimum effort levels, this becomes:
\[
\E \left[ p_q^k (\vec{x}, \vec{y} ) \right] =c_q^k - d_q^k \left( \sum_{i \in \Sc} \xi_{q,i}^k \sigma_i^2(e_i^*) \right) =
\]
\[
c_q^k - d_q^k \left( \sum_{i \in \Sc}  \frac{\xi_{q,i}^k}{2 \left( \sum_{j \in \Bc}d_i^j\right) \alpha_i} \right)
\]

Also using the expression for $e_i^*$ given in \eqref{eq:specific_effort}, buyer $k$ has the following optimization problem:
{\scriptsize 
\begin{eqnarray}
\label{eq:opt_p}
\min_{c^k,d^k}~
&
\sum_{i \in \Sc} \frac{ \gamma_i^k }{2 d_i^{total} \alpha_i} + 
\eta^k \sum_{i \in \Sc} \left( c_i^k - d_i^k \left[ \sum_{l \in \Sc} \frac{ \xi_{i,l}^k }{2 d_l^{total} \alpha_l} \right] \right) 
\\
\label{eq:const_total_IR_final}
\text{s.t.} 
& \sum_{j \in \Bc} \left[ c_i^j - d_i^j \left( \sum_{l \in \Sc}  \frac{\xi_{i,l}^j}{2 d_l^{total} \alpha_l} \right) \right] - \frac{\ln\left(2 d_i^{total} \alpha_i\right)}{2 \alpha_i} \geq 0 \\
\label{eq:const_single_IR_final}
& c_i^k - d_i^k \left( \sum_{l \in \Sc}  \frac{\xi_{i,l}^k}{2 d_l^{total} \alpha_l} \right) \geq 0 \\
\label{eq:d_tot}
& d_i^{total} = \sum_{j \in \Bc}d_i^j \\
& c_i^k \geq 0, d_i^k \geq 0
\end{eqnarray}
}
Every constraint above holds for all $i \in \Sc$. Here, \eqref{eq:d_tot} is a definitional, rather than binding, constraint. Also, note that without loss of generality, we can take $\eta^k = 1$, by normalizing the $\gamma_i^k$ accordingly. Additionally, we can remove the constraint $c_i^k \geq 0$, as it is redundant in light of the constraint $c_i^k - d_i^k \left( \sum_{l \in \Sc}  \frac{\xi_{i,l}^k}{2 d_l^{total} \alpha_l} \right) \geq 0$, since $\xi \geq 0$ and $d \geq 0$.

 This leads to the following result.
\begin{theorem}
\label{thm:nash}
Consider the game where each buyer's objective is to solve the optimization in \eqref{eq:opt_p}, and assume $\gamma_i^j \geq 0$ for all $i \in \Sc, j \in \Bc$. Then there are either an infinite number of generalized Nash equilibria, or there is no generalized Nash equilibrium.

Furthermore, in the case where there are an infinite number of generalized Nash equilibria, there is a unique collection of $d$ parameters, in the sense that if $(\vec{c},\vec{d})$ and $(\vec{c'},\vec{d'})$ are both generalized Nash equilibria, then $\vec{d} = \vec{d'}$. Additionally, the $c$ parameters lie in the convex polytope defined by the following constraints:
\[
\sum_{j \in \Bc} c_i^j = \sum_{j \in \Bc}d_i^j \left( \sum_{l \in \Sc}  \frac{\xi_{i,l}^j}{2 d_l^{total} \alpha_l} \right) + \frac{\ln\left(2d_i^{total} \alpha_i\right)}{2 \alpha_i}
\]
\[
c_i^k \geq d_i^k \left( \sum_{l \in \Sc}  \frac{\xi_{i,l}^k}{2 d_l^{total} \alpha_l} \right)
\]
The effort exerted by each data source is the same in each generalized Nash equilibrium.
\end{theorem}

Before proving this theorem, we discuss the assumption that $\gamma_i^j \geq 0$. This implies that, for each data buyer, the penalty for other data buyer's successful estimation does not outweigh the benefit of having a good estimator. This assumption means that no data buyer will have incentive to drive the variance of one data source up towards infinity.

We prove the following useful lemma, and then prove our theorem.

\begin{lemma}
\label{lem:datam_eq_holds}
Suppose $(\vec{c},\vec{d})$ is a GNE for the game defined by \eqref{eq:opt_p}. The following equality holds for all $i$ and $k$:
\[
c_i^k = \sum_{j \in \Bc}d_i^j \left( \sum_{l \in \Sc}  \frac{\xi_{i,l}^j}{2 d_l^{total} \alpha_l} \right) + \frac{\ln\left(2 d_i^{total} \alpha_i\right)}{2 \alpha_i} -
\sum_{j \in -k} c_i^j
\]
In other words, \eqref{eq:const_total_IR_final} is always tight in equilibrium.
\end{lemma}
\begin{proof}
To prove this, note that, by the cost function of buyer $k$, $c_i^k$ will always be chosen such that at least one of \eqref{eq:const_total_IR_final} and \eqref{eq:const_single_IR_final} is tight. Suppose \eqref{eq:const_single_IR_final} is exclusively active, i.e. 
\[
c_i^k - d_i^k \left( \sum_{l \in \Sc}  \frac{\xi_{i,l}^k}{2 d_l^{total} \alpha_l} \right) = 0
\]
\begin{equation}
\label{eqn:IR_ineq_const}
c_i^k > \sum_{j \in \Bc}d_i^j \left( \sum_{l \in \Sc}  \frac{\xi_{i,l}^j}{2 d_l^{total} \alpha_l} \right) + \frac{\ln\left(2 d_i^{total} \alpha_i\right)}{2 \alpha_i} -
\sum_{j \in -k} c_i^j
\end{equation}
Note that \eqref{eqn:IR_ineq_const} is the same constraint for every data buyer. In other words, if it is loose for $k$, it is loose for all other $j$. Thus, some other buyer $j$ can reduce their $c_i^j$ and lower their cost, and thus $(\vec{c},\vec{d})$ cannot be an equilibrium.

This argument does fall apart in one situation, however. No buyer can reduce their cost just by modifying $c$ if \eqref{eq:const_single_IR_final} is tight for all buyers $k$, i.e. for all $k$:
\[
c_i^k - d_i^k \left( \sum_{l \in \Sc}  \frac{\xi_{i,l}^k}{2 d_l^{total} \alpha_l} \right) = 0
\]
In this case, \eqref{eq:const_total_IR_final}, which we assumed held loosely, becomes $2 d_i^{total} \alpha_i < 1$. Let buyer $k$ increase $d_i^k$ such that $2 d_i^{total} \alpha_i = 1$, and then choose a new $c^k$ such that \eqref{eq:const_single_IR_final} holds tightly, i.e. $c^k = d_i^k \left( \sum_{l \in \Sc}  \frac{\xi_{i,l}^k}{2 d_l^{total} \alpha_l} \right)$. Note that this decreases their cost:
\[
\sum_{i \in \Sc} \frac{ \gamma_i^k }{2 d_i^{total} \alpha_i} > \sum_{i \in \Sc} \gamma_i^k
\]
(This uses the fact that, since \eqref{eq:const_single_IR_final} holds for all buyers $j$, the second term disappears.) Additionally, all the constraints of the original optimization are still satisfied, so $(c_i^k, d_i^k)$ was not an optimizer for buyer $k$.

This concludes our proof.
\end{proof}
\begin{proof} \emph{ (Theorem~\ref{thm:nash}) }
We invoke Lemma~\ref{lem:datam_eq_holds} and substitute this into the objective function, \eqref{eq:opt_p}, for buyer $k$. Let:
\[
J\left(\vec{c},\vec{d}\right) = \sum_{i \in \Sc} \Bigg( \frac{ \gamma_i^k }{2 d_i^{total} \alpha_i} + 
\]
\[
\sum_{j \in -k} \left( d_i^j \left( \sum_{l \in \Sc}  \frac{\xi_{i,l}^j}{2 d_l^{total} \alpha_l} \right) - c_i^j\right)+ \frac{\ln\left(2 d_i^{total} \alpha_i\right)}{2 \alpha_i} \Bigg)
\]
This yields:
{
\begin{eqnarray}
\nonumber
\min_{c^k,d^k}~
&
J\left(\vec{c},\vec{d}\right)
\\ 
\nonumber
\text{subject to } 
& c_i^k - d_i^k \left( \sum_{l \in \Sc}  \frac{\xi_{i,l}^k}{2 d_l^{total} \alpha_l} \right) \geq 0 \\
\nonumber
& d_i^{total} = \sum_{j \in \Bc}d_i^j \\
\nonumber
& d_i^k \geq 0
\end{eqnarray}
}
We quickly manipulate the cost function a little to a more desirable form:
\[
\sum_{i \in \Sc} \Bigg( \frac{ \gamma_i^k }{2 d_i^{total} \alpha_i} + 
\]
\[
\sum_{j \in -k} \left( d_i^j \left( \sum_{l \in \Sc}  \frac{\xi_{i,l}^j}{2 d_l^{total} \alpha_l} \right) - c_i^j\right)+ \frac{\ln\left(2 d_i^{total} \alpha_i\right)}{2 \alpha_i} \Bigg) =
\]
\[
\sum_{i \in \Sc} \left( \frac{ \gamma_i^k }{2 d_i^{total} \alpha_i}+ \frac{\ln\left(2 d_i^{total} \alpha_i\right)}{2 \alpha_i} \right)
+ 
\]
\[
\sum_{i \in \Sc} \sum_{j \in -k}\sum_{l \in \Sc} \frac{ d_i^j \xi_{i,l}^j}{2 d_l^{total} \alpha_l}  -
\sum_{i \in \Sc} \sum_{j \in -k}c_i^j =
\]
\[
\sum_{i \in \Sc} \left( \frac{ \gamma_i^k }{2 d_i^{total} \alpha_i}+ \frac{\ln\left(2 d_i^{total} \alpha_i\right)}{2 \alpha_i} \right)
+
\]
\[
\sum_{i \in \Sc} \sum_{j \in -k}\sum_{l \in \Sc} \frac{ d_l^j \xi_{l,i}^j}{2 d_i^{total} \alpha_i}  -
\sum_{i \in \Sc} \sum_{j \in -k}c_i^j =
\]
\[
\sum_{i \in \Sc} \Bigg( \frac{ \gamma_i^k }{2 d_i^{total} \alpha_i}+ 
\]
\[
\sum_{j \in -k} \left( \sum_{l \in \Sc} \frac{ d_l^j \xi_{l,i}^j}{2 d_i^{total} \alpha_i} - c_i^j \right)+
\frac{\ln\left(2 d_i^{total} \alpha_i\right)}{2 \alpha_i}\Bigg)
\]
Note the index swap on the $\xi$ terms in the second equality. Then 
define:
\[
J_i^k(d_i^k, c^{-k}, d^{-k}) = 
\frac{ \gamma_i^k }{2 d_i^{total} \alpha_i}+ 
\]
\[
\sum_{j \in -k} \left( \sum_{l \in \Sc} \frac{ d_l^j \xi_{l,i}^j}{2 d_i^{total} \alpha_i} - c_i^j \right)+
\frac{\ln\left(2 d_i^{total} \alpha_i\right)}{2 \alpha_i}
=
\]
\[
\frac{ \gamma_i^k +  \sum_{j \in -k}\sum_{l \in \Sc}d_l^j \xi_{l,i}^j }{2 d_i^{total} \alpha_i}-\sum_{j \in -k}  c_i^j +
\frac{\ln\left(2 d_i^{total} \alpha_i\right)}{2 \alpha_i}
\]

Thus, the overall optimization can again be re-written:
\begin{eqnarray}
\nonumber
\min_{c^k,d^k}~
&
\sum_{i \in \Sc} 
J_i^k(d_i^k,c^{-k},d^{-k})
\\ 
\nonumber
\text{subject to } 
& c_i^k - d_i^k \left( \sum_{l \in \Sc}  \frac{\xi_{i,l}^k}{2 d_l^{total} \alpha_l} \right) \geq 0 \\
\nonumber
& d_i^{total} = \sum_{j \in \Bc}d_i^j \\
\nonumber
& d_i^k \geq 0
\end{eqnarray}
We differentiate the cost with respect to $d_q^k$:
\[
\frac{\partial}{\partial d_q^k} \sum_{i \in \Sc} 
J_i^k(d_i^k,c^{-k},d^{-k}) =
\frac{\partial}{\partial d_q^k} J_q^k(d_q^k,c^{-k},d^{-k}) =
\]
\[
- \frac{ \gamma_q^k +  \sum_{j \in -k}\sum_{l \in \Sc}d_l^j \xi_{l,q}^j }{2 (d_q^{total})^2 \alpha_q}+ \frac{1}{2 d_q^{total} \alpha_q} =
\]
\[
\frac{ -\gamma_q^k - \sum_{j \in -k}\sum_{l \in \Sc}d_l^j \xi_{l,q}^j  + d_q^{total} }{2 (d_q^{total})^2 \alpha_q} =
\]
\[
\frac{ -\gamma_q^k -\sum_{j \in -k}\sum_{l \in -q}d_l^j \xi_{l,q}^j + d_q^k }{2 (d_q^{total})^2 \alpha_q}
\]
Note that we use the fact that $\xi_{q,q}^j = 1$ for all $j$. 
It is easy to see that:
{\scriptsize 
\[
\frac{\partial}{\partial d_q^k} J_q^k(d_q^k,c^{-k},d^{-k})
\begin{cases}
< 0 & \text{if } 0 \leq d_q^{k} < \gamma_q^k + \sum_{j \in -k}\sum_{l \in -q}d_l^j \xi_{l,q}^j \\
= 0 & \text{if } d_q^{k} = \gamma_q^k + \sum_{j \in -k}\sum_{l \in -q}d_l^j \xi_{l,q}^j \\
> 0 & \text{if } d_q^{k} > \gamma_q^k + \sum_{j \in -k}\sum_{l \in -q}d_l^j \xi_{l,q}^j 
\end{cases}
\]
}
Thus, the maximizing $d_q^k$ is given by:
\begin{equation}
\label{eq:optimal_d}
d_q^k = \gamma_q^k + \sum_{j \in -k}\sum_{l \in -q}d_l^j \xi_{l,q}^j
\end{equation}
Performing this analysis for all combinations of $q \in \mathcal{S}$ and $k \in \mathcal{B}$  yields a system of $M\times N$ equations  with $M \times N$ unknowns, of the form \eqref{eq:optimal_d}.

As we have before, let $\vec{d}$ denote a column vector with entries $d_i^j$ for each $i \in \Sc$ and $j \in \Bc$. Similarly, let $\vec{\gamma}$ denote a column vector containing all the terms of the form $\gamma_i^j$. Then,
we may represent this system of equations with the following matrix equation:
\begin{equation}
\label{eq:eq_conditions}
\vec{d} = A \vec{d} + \vec{\gamma}
\end{equation} 
Here, $A$ is a non-negative matrix whose entries are the values of the various $\xi$ parameters at the appropriate places, such that \eqref{eq:eq_conditions} expresses the set of equality constraints defined by \eqref{eq:optimal_d} for all $q \in \Sc$ and $k \in \Bc$. 
To find an GNE of this game, it suffices to find a solution to \eqref{eq:eq_conditions} such that $d_i^j \geq 0$ for all $i$ and $j$.

Systems of equations of this form are well studied in the economics literature, as they are of the form specified by the celebrated Leontief input-output model. It has been shown that such systems of equations have a non-negative solution if and only if $\rho(A) < 1$, where $\rho(A)$ is the spectral radius of $A$~\cite{Stanczak2006}. Moreover, if such a solution exists, it must be unique.

Thus, if $\rho(A) < 1$, inversion of this $A$ matrix yields the equilibrium $\vec{d}$, and, by Lemma~\ref{lem:datam_eq_holds}, we can pick any $\vec{c}$ that satisfies:
\[
\sum_{j \in -\Bc} c_i^j = \sum_{j \in \Bc}d_i^j \left( \sum_{l \in \Sc}  \frac{\xi_{i,l}^j}{2 d_l^{total} \alpha_l} \right) + \frac{\ln\left(2 d_i^{total} \alpha_i\right)}{2 \alpha_i} 
\]
\[
c_i^k \geq d_i^k \left( \sum_{l \in \Sc}  \frac{\xi_{i,l}^k}{2 d_l^{total} \alpha_l} \right)
\]

If $\rho(A) \geq 1$, there will not exist a non-negative solution and there is no point $(\vec{c},\vec{d})$ that simultaneously optimizes \eqref{eq:opt_p} for all $k$. It follows that there is either a unique $\vec{d}$ that will constitute a Nash solution for the game, which produces a convex polytope of potential GNE, or there this no solution to the game, as desired.
\end{proof}

It is interesting to note that the existence of GNE depends solely on the value of the $\xi$ parameters; it does not depend on the magnitude of the $\gamma$ parameters. This implies that the existence or non-existence of GNE in this game is simply an artifact of the incentive mechanisms we have chosen to analyze, and does not depend on whether or not there are solutions that are beneficial to all parties involved. Note that we chose this incentive mechanism based on several desirable properties in the single-buyer case; whether or not there exist mechanisms that extend to multi-buyer games in a fashion that provides good efficiency properties is an open problem that we are currently investigating. 
In Section~\ref{sec:datam_examples}, we calculate the $\xi$ parameters for a specific example, and see how equilibrium solutions in these marketplaces collapse as the characteristics are varied.

Additionally, note that, in the case where there is a continuum of GNE, the effort exerted by data sources and $\vec{d}$ parameters are the same across all equilibria. The ambiguity arises in the $\vec{c}$ parameters. In other words, the ambiguity arises in determining which data buyers will pay to ensure that each data source's total compensation covers the cost of their effort. 
In the extreme case, it is possible for one firm to pay for the entirety of the expected compensation offered to the data sources, while the the firms pay nothing on expectation. That is, for some $k \in \mathcal{B}$, $\sum_{i \in \mathcal{S}} p_i^k(\vec{x}, \vec{y}) = \sum_{i \in \mathcal{S}} e_i^*$, and for all $j \neq k $,  $\sum_{i \in \mathcal{S}} p_i^j(\vec{x}, \vec{y}) = 0$. 
In Section~\ref{sec:datam_discussion} we discuss possible mechanisms to alleviate the disparity that may arise in these situations.

We next turn to analyzing the total utility experienced in the marketplace for a given outcome of the game. We begin with the following definition.
\begin{definition}
\emph{(Ex-ante social loss of the data market)} 
Suppose that $\eta^j = 1$ for all buyers $j$. 
Let $\vec{e}$ be the vector denoting the level of effort the data sources exert. Then, we define the \emph{ex-ante social loss} the marketplace to be the sum of the utility functions of all the data buyers and data sources:
\[
\mathcal{L}(\vec{e}) = \sum_{j \in \mathcal{B}} \Bigg(
\E \Bigg[ 
\left( \hat{f}^j_{\vec{x},\vec{y^j}}(x^*) - f(x^*) \right)^2 -
\]
\[
\sum_{k \in -j} \delta_k^j \left( \hat{f}^k_{\vec{x},\vec{y^k}}(x^*) - f(x^*) \right)^2
\Bigg]
\Bigg) + \sum_{i \in \Sc} e_i
\]
\end{definition}
Note that this sum does not include any of the payments made in the marketplace, as they are simply lossless transfers of wealth. We require the additionally assumption that $\eta^j = 1$ for all buyers $j$ to ensure that these transfers of wealth are lossless from a utility perspective, i.e. the buyers and sources value the payment equally. This assumption allows us to isolate the social loss due to the mechanism, and ignore any losses due to differential preferences in payment currency.

\begin{theorem}
\label{th:always_lose}
Suppose that Assumptions~\ref{ass:sigma_form} and~\ref{ass:separable} hold. Further, assume that $\gamma_i^j> 0$, for all $i \in \mathcal{S}$ and $ j \in \mathcal{B}$, and that $\xi_{i,l}^j > 0$ for some $i, l \in \mathcal{S}, \ j \in \mathcal{B}$. Finally, suppose GNE solutions exist for the game, and let $\vec{e^*}$ denote the unique level of effort exerted by the data sources across each of these GNE solutions, as stipulated by Theorem~\ref{thm:nash}. Then, there exists $\vec{\hat{e}} \in \R^N$ such that $\mathcal{L}\left(\vec{\hat{e}}\right)< \mathcal{L}\left(\vec{e^*}\right)$. Furthermore, the socially optimal levels of effort, $\vec{\hat{e}}$, are always less than the induced levels of effort at equilibrium $\vec{e^*}$.
\end{theorem}
\begin{proof}
We begin by calculating solving for the value of ${\vec{\hat e}}$ which minimizes the value of $\mathcal{L}\left(\vec{\hat{e}}\right)$. Invoking Assumption~\ref{ass:separable} and our definition of $\gamma$, we may write:
\begin{equation}
\mathcal{L}\left(\vec{\hat{e}}\right) = \sum_{i \in \mathcal{S} }\sum_{j \in \mathcal{B}} \gamma_i^j \sigma_i^2(\hat{e_i}) + \sum_{i \in \mathcal{S}} \hat{e}_i  
\end{equation}

Taking the derivative with respect to $\hat{e_q}$ and repeating our analysis with Assumption~\ref{ass:sigma_form}, and setting the resulting equation to zero we obtain:
\[
-2 \alpha_q \left( \sum_{j \in \mathcal{B}}\gamma_q^j \right) \exp(-2\alpha_q \hat{e}_q)  + 1 = 0
\]
We can re-arrange this to yield:
\begin{equation}
\label{eq:optimal_effort_social}
\hat{e}_q =  \frac{\ln(2\alpha_q \sum_{j \in \mathcal{B}} \gamma_q^j)}{2\alpha_q}
\end{equation}

Note, that $\mathcal{L}$ is strictly convex with respect to $\hat e_q$, so choosing the entries of $\hat{\vec{e}}$ according to equation \eqref{eq:optimal_effort_social} must produce the unique minimizer of $\mathcal{L}\left(\vec{\hat{e}}\right)$.

Next we compare this to the level of effort the sources produce in the GNE of the game between the buyers. By \eqref{eq:optimal_d}, we obtain that in the GNE for all $i \in \Sc$ and $k \in \Bc$:
\[
d_i^k = \gamma_i^k + \sum_{j \in -k} \sum_{l \in -i} d_l^j \xi_{l,i}^j \geq \gamma_i^k > 0
\]
Furthermore, since there exists at least one $\xi > 0$, we know that $d_q^k > \gamma_q^k$ for some data source $q$ and buyer $k$. It follows that, for this particular $q$:
\[
\sum_{j \in \Bc} d_q^j > \sum_{j \in \Bc} \gamma_q^j
\]
Thus, by \eqref{eq:specific_effort}, we see that
\begin{equation}
e_q^* =  \frac{\ln(2\alpha_q \sum_{j \in \mathcal{B}} d_q^j)}{2\alpha_q} > \frac{\ln(2\alpha_q \sum_{j \in \mathcal{B}} \gamma_q^j)}{2\alpha_q} = \hat{e}_q.
\end{equation}

Thus the theorem is proved, since it must be the case that $\mathcal{L}\left(\vec{\hat{e}}\right)< \mathcal{L}\left(\vec{e^*}\right)$ since we chose $\vec{\hat{e}}$ to be the unique minimizer of $\mathcal{L}$.
\end{proof}

Theorem~\ref{th:always_lose} shows that there is always some social loss, ex-ante, from a Nash solution compared to the social optimum. Furthermore, the proof provides a way to identify where this loss is incurred, and how to calculate how much is lost. Note that the social welfare is always lost because the effort induced in equilibrium is higher than is socially optimal. This captures the intuition that each data buyer has a negative externality: they wish to improve their estimates without considering how their improved estimates hurt other data buyers.

The proof itself also provides strong intuition on the $\xi$ parameters between buyers. Loosely speaking, these $\xi$ parameters can be thought of as a measure of each buyer's `market power', in the sense that it quantifies how much one buyer can influence the payment contracts of other buyers in the data market to his advantage. 
As an extremal case, when $\xi_{i,l}^j = 0$ for all $i,l \in \Sc$ and $j \in \Bc$, there is no coupling between the payments the buyers make, and the social optimum coincides with the Nash solution.

\section{Example: Between two firms}
\label{sec:datam_examples}

In this section, we present an example which demonstrates how a data market may
collapse as the parameters of the system are varied. This example will  also
demonstrate how the efficiency of the data market, in terms of the ex-ante social loss
function $\mathcal{L}$, changes as the market approaches this collapse. In
particular, we consider the case where there are two data sources ($s_1$ and
$s_2$) and two firms acting as data buyers ($b_1$ and $b_2$). Each of the data sources is capable of estimating the function 
\begin{equation}
f \colon [-1, \ 1] \to \R.
\end{equation}

Let $x_1$, $x_2 \in [-1, 1]$ denote the locations where $s_1$ and $s_2$ sample $f$, respectively. Assume that each of the data sources are as defined in Assumption~\ref{ass:sigma_form}, with the characteristic parameters $\alpha_1 = \alpha_2 = 1$.

Next, we assume that each of the data buyers is performing linear regression on $f$, using the samples reported by the data source. In this case~\cite{cai:2015aa}:
\[
g_j(\vec{x}, F_j, \sigma^2(\vec{e})) =
\E_{x^* \sim F_j} \Bigg[ 
\begin{bmatrix}
x^* \\
1
\end{bmatrix}^\T (X^\T X)^{-1} X^\T \cdot 
\]
\[
\diag(\sigma_1^2(e_1), \sigma_2^2(e_2)) \cdot 
X (X^\T X)^{-1} 
\begin{bmatrix}
x^* \\
1
\end{bmatrix}
\Bigg] =
\]
\[
\gamma_1^j \sigma_1^2(e_1) + \gamma_2^j\sigma_2^2(e_2)
\]
In this example, we assume $F_1 = F_2$ as the uniform distribution on the domain of $f$, $[-1,1]$. Thus, for $i \in \{1,2\}$:
\[
\gamma_i^1 = \gamma_i^2 = \frac{(x_1 - x_2)^2/3 + (x_i^2 - x_1x_2)^2}{(x_1^2+x_2^2-2x_1x_2)^2}
\] 
Note that, by these assumptions, $g_1 = g_2$, and furthermore:
\[
\xi_{1,2}^1 = \xi_{1,2}^2 = g(x_2, \delta_{x_1}, \sigma_2^2(e_2)) = \frac{(x_1x_2+1)^2}{(x_2^2+1)^2}
\]
\[
\xi_{2,1}^1 = \xi_{2,1}^2 = g(x_1, \delta_{x_2}, \sigma_1^2(e_1)) = \frac{(x_1x_2+1)^2}{(x_1^2+1)^2}
\]
For illustrative purposes, we fix $x_2 = 1$, and see what happens to the data market as we vary $x_1$ along the interval $[-1,1]$. 

Note that, when $x_1 = x_2 = 1$, it is no longer possible to construct a linear estimator of $f$ because there is insufficient data. Thus, the example shows how the game between buyers behaves as it becomes increasingly difficult to construct good estimators. The $\vec{d}$ parameters of any Nash solution can be found by solving:
\begin{equation}
\label{eq:matrix_eq}
\begin{bmatrix}
\gamma_1^1 \\
\gamma_1^2 \\
\gamma_2^1 \\
\gamma_2^2
\end{bmatrix} =
\underbrace{
\begin{bmatrix}
1 & 0 & 0 & -\xi_{2,1}^2 \\
0 & 1 & -\xi_{2,1}^1 & 0 \\
0 & -\xi_{1,2}^2 & 1 & 0 \\
-\xi_{1,2}^1 & 0 & 0 & 1
\end{bmatrix}
}_{B}
\begin{bmatrix}
d_1^1 \\
d_1^2 \\
d_2^1 \\
d_2^2
\end{bmatrix}
\end{equation}
Note that this $B$ matrix is equal to $I - A$ as defined in the proof. We numerically solve this system of equations for varying values of $x_1 \in [-1,1]$, and the results are shown in Figures~\ref{fig:gamma_params} through~\ref{fig:dm_loss}.

Figures~\ref{fig:gamma_params} and~\ref{fig:xi_params}, demonstrate how the
$\gamma$ and $\xi$ parameters of the game change as a function of $x_1$. Figure~\ref{fig:d_params} demonstrates the d parameters that the buyers will offer the data sources as  $x_1$ varies. And finally, Figure \label{fig:loss} demonstrates
the price of anarchy in the data market, as a function of $x_1$ which is given
by:
\[
\frac{\mathcal{L}\left(\vec{e^*}\right)}{\mathcal{L}\left(\vec{\hat{e}}\right)}
\]
Here, $\vec{e}^*$ is the induced effort of the sensors in the Nash solution of the game between data buyers, and $\vec{\hat{e}}$ is the socially optimal effort for data sources to exert. Further comments in the captions of Figures~\ref{fig:d_params} and~\ref{fig:dm_loss} demonstrate the inefficiencies that arise in this example.

\begin{figure}[h]\centering
\includegraphics[width=0.35\textwidth]{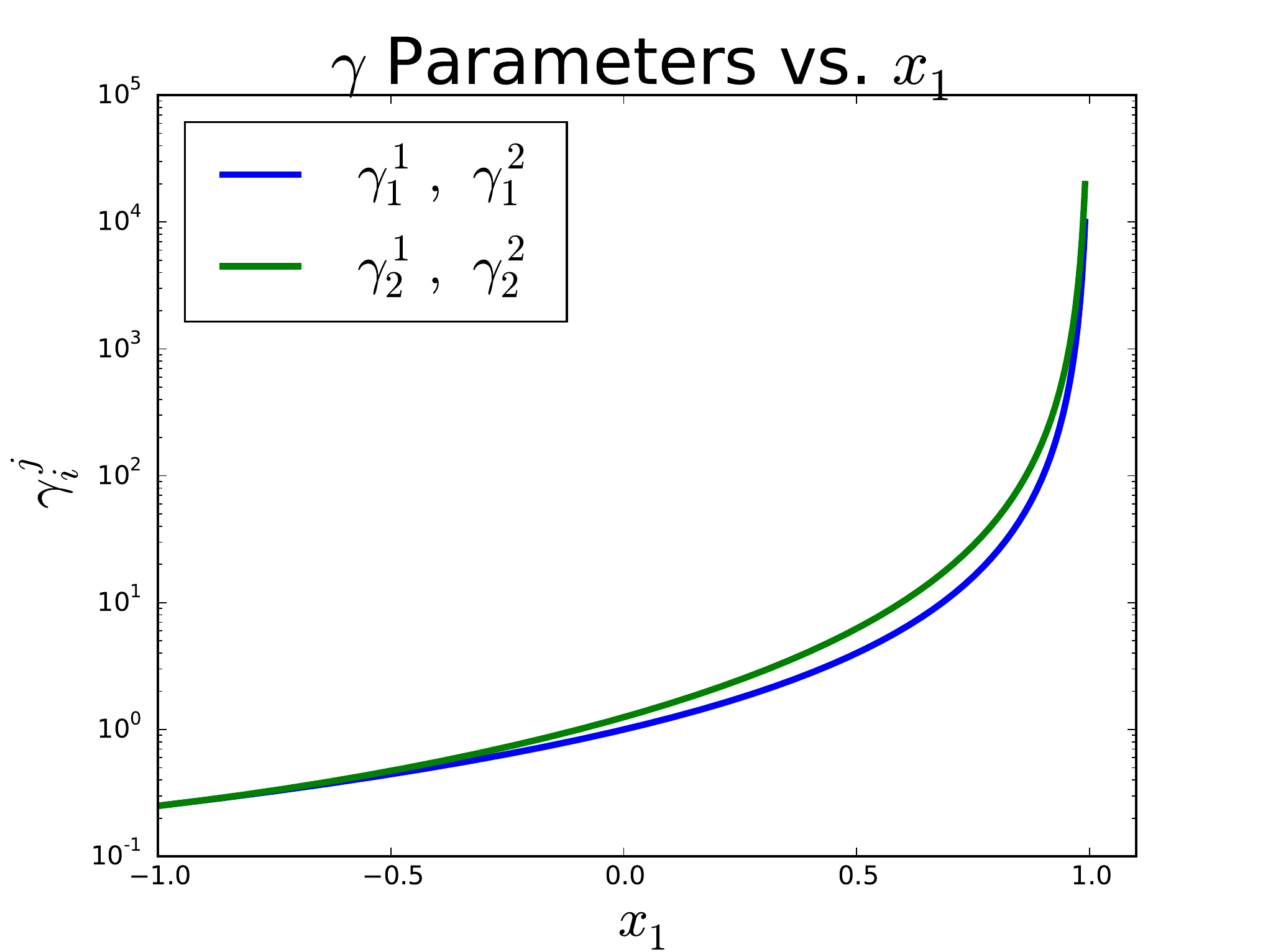}
\caption{This figure depicts how the various $\gamma$ parameters of the system
vary as a function of $x_1$. Note that as $x_1 \to 1$, $\gamma$ diverges to
infinity, which reflects the fact that as $x_1$ and $x_2$ become increasingly
close it becomes more difficult to generate a linear estimator from samples at
these data points. }\label{fig:gamma_params}
\end{figure}
\begin{figure}[h]\centering
\includegraphics[width=0.35\textwidth]{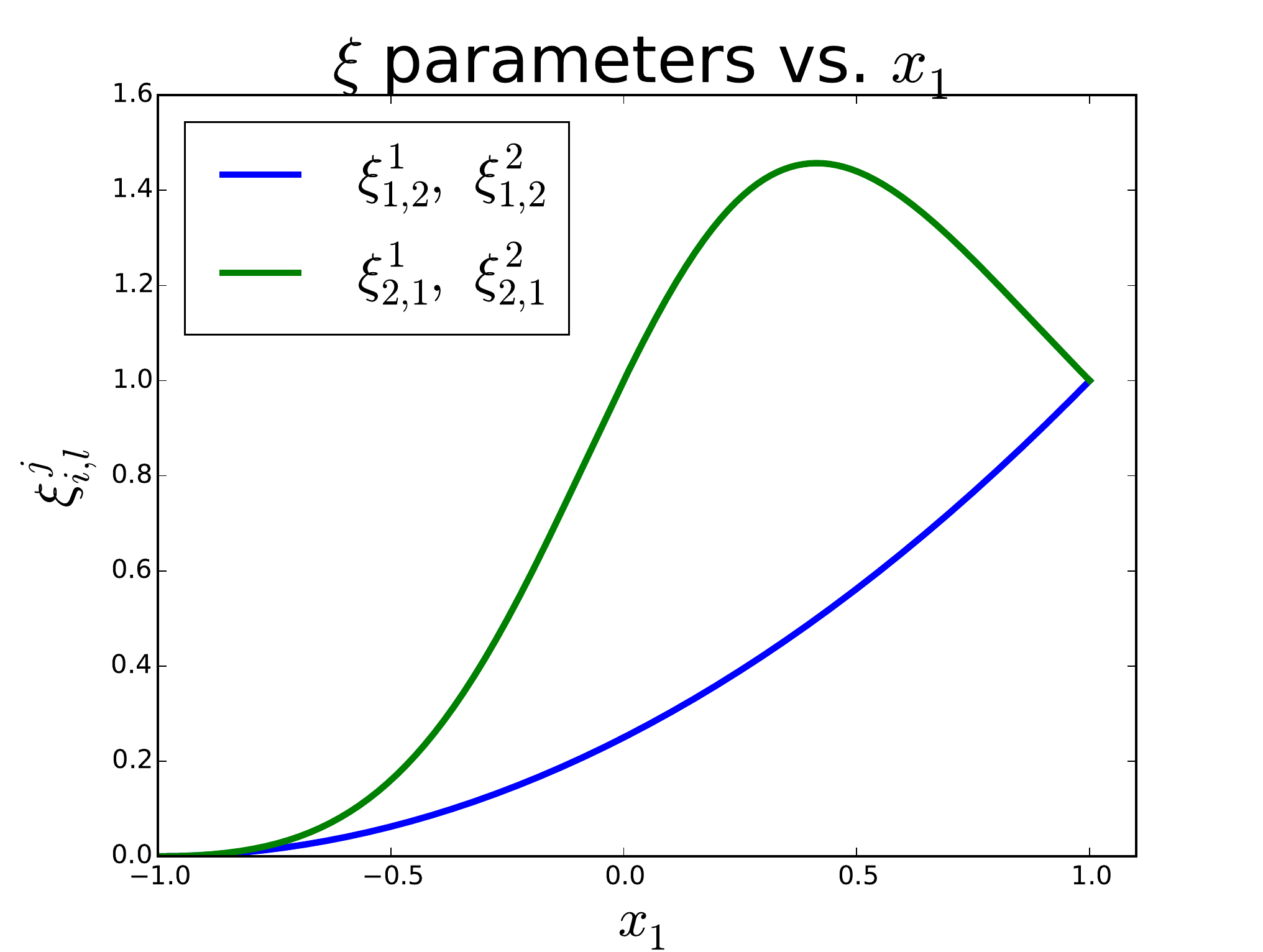}
\caption{This figure depicts how the various $\xi$ parameters vary as a function
    of $x_1$. As all of the $\xi$ parameters converge to a value of $1$ as $x_1
    \to 1$, the matrix $B$ in equation \eqref{eq:matrix_eq} becomes singular,
    causing the breakdown of solutions for the $d$ parameters, as is depicted in
    Figure~\ref{fig:d_params}. }\label{fig:xi_params}
\end{figure}
\begin{figure}[h]\centering
\includegraphics[width=0.35\textwidth]{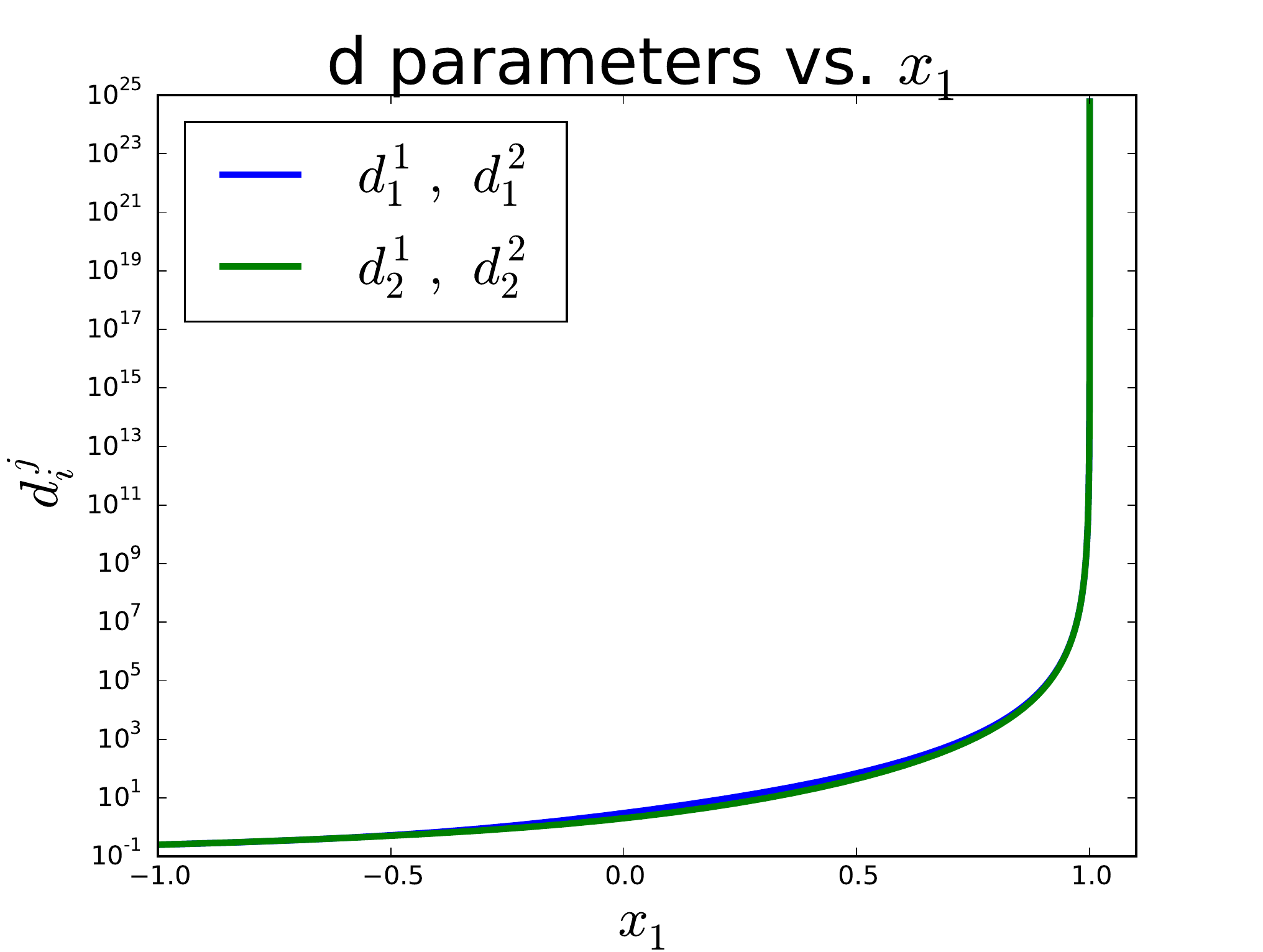}
\caption{This figure depicts the Nash equilibrium $\vec{d}$ parameters for the game
    between the buyers as a function of $x_1$. Note that, as $x_1 \to 1$, the
    $\vec{}d$
    parameters go off to infinity, and the Nash equilibria between the buyers
    breaks down. Comparing these results to Figure~\ref{fig:gamma_params}, we
see that the $\vec{d}$ parameters diverge much more quickly than the $\gamma$
parameters, meaning that in the Nash equilibria to the game between the two buyers becomes increasingly inefficient as $x_1 \to 1$. }
\label{fig:d_params}
\end{figure}
\begin{figure}[h]
    \centering
\includegraphics[width=0.35\textwidth]{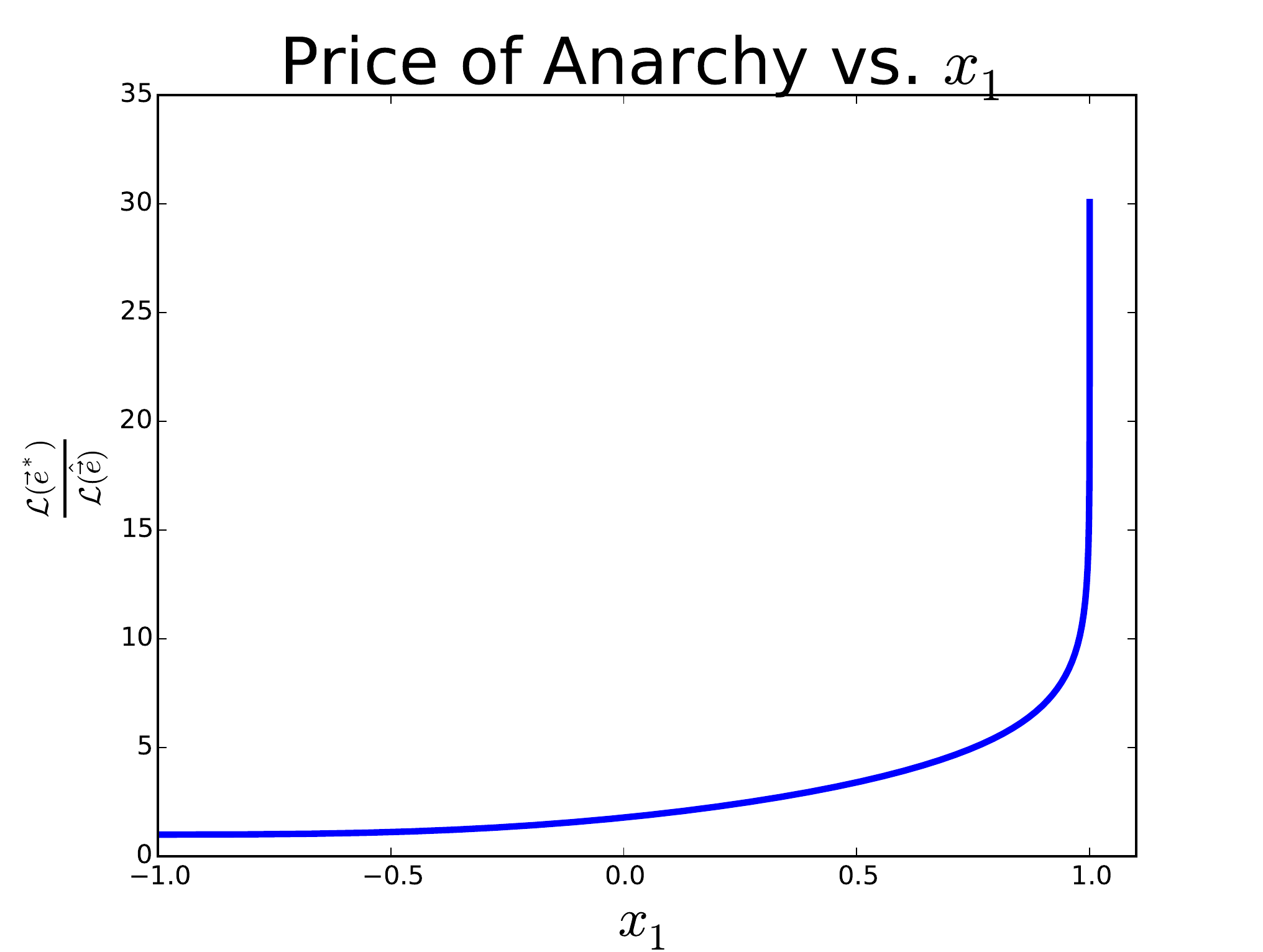}
    \caption{This figure depicts the price of anarchy for the marketplace as a function of $x_1$. When $x_1=-1$, the payments in the marketplace are decoupled and the $\xi$ parameters  are all zero; in this instance the price of anarchy is 1, and the market is perfectly efficient. However, as $x_1 \to 1$, the price of anarchy diverges asymptotically to infinity, and the marketplace becomes increasingly inefficient as it becomes more difficult for the buyers to construct the estimators they desire. }
    \label{fig:dm_loss}
\end{figure}

\section{Closing remarks}
\label{sec:datam_discussion}

We've analyzed the game that forms between a set of data buyers when they wish to communally incentivize a collection of strategic data sources, using a mechanism that has been proposed in the literature. We derived, for a particular form of the game, conditions for the existence of GNE, and demonstrated that these solutions are frequently socially inefficient. This motivates future work to develop a richer class of incentive mechanisms which alleviate these issues. Possible solutions include more complex pricing mechanisms, or perhaps the addition of a trusted third party market-maker to mediate socially beneficial transactions in these data markets.

\bibliographystyle{IEEEtran}
\bibliography{DONG_ROY_refs}

\end{document}